\documentclass[conference,12pt]{article} 

\usepackage{amsmath}
\usepackage{amssymb}
\usepackage{algorithm}
\usepackage{amsthm}
\usepackage{graphicx}
\usepackage{subfigure}
\usepackage{enumerate}
\usepackage{color}

\newtheorem{thm}{Theorem}[section]
\newtheorem{lem}{Lemma}[section]
\newtheorem{cor}{Corollary}[section]

\newcommand{\x}{\mathbf{x}}

\newcommand{\A}{\mathbf{a}}

\newcommand{\eps}{\epsilon}

\begin{document}

\title{Reconstruction of Integers from Pairwise Distances}
\author{\begin{tabular}[t]{c@{\extracolsep{5em}}c}
Kishore Jaganathan & Babak Hassibi\end{tabular}\\
 \\
        Department of Electrical Engineering, Caltech \\
       Pasadena, CA~~91125
}
\date{}
\maketitle

\begin{abstract} 
Given a set of integers, one can easily construct the set of their pairwise distances. We consider the inverse problem: given a set of pairwise distances, find the integer set which realizes the pairwise distance set. This problem arises in a lot of fields in engineering and applied physics, and has confounded researchers for over 60 years. It is one of the few fundamental problems that are neither known to be NP-hard nor solvable by polynomial-time algorithms. Whether unique recovery is possible also remains an open question.

In many practical applications where this problem occurs, the integer set is naturally sparse (i.e., the integers are sufficiently spaced), a property which has not been explored. In this work, we exploit the sparse nature of the integer set and develop a polynomial-time algorithm which provably recovers the set of integers (up to linear shift and reversal) from the set of their pairwise distances with arbitrarily high probability if the sparsity is $O(n^{1/2-\eps})$. Numerical simulations verify the effectiveness of the proposed algorithm.
\end{abstract}

\section{Introduction}

We consider the problem of reconstructing a set of integers from the set of their pairwise distances. For example, consider the set $V=\{2,5,13,31,44\}$. Its pairwise distance set is given by $W=\{0,3,8,11,13,18,26,29,31,39,42\}$. We look at the problem of recovering the integer set $V$ from the pairwise distance set $W$\footnote[2]{If $V$ has a pairwise distance set $W$, then sets $c\pm V$ also have the same pairwise distance set $W$ for any integer $c$. These solutions are considered equivalent, and in all the applications it is considered good enough if any equivalent solution, i.e., up to linear translation and flipping, is recovered.}.

This recovery problem dates back to the origins of the classical phase retrieval problem in the 1930s \cite{patt1,patt2} and has received a lot of attention from researchers. More recently, it has arisen in computational biology, specifically in restriction site mapping of DNA \cite{stef}. This problem has also been posed as a computational geometry problem \cite{shamos}.


\subsection{Phase Retrieval}

Many measurement systems in practice can output only the squared-magnitude of the Fourier transform. Phase information is completely lost, because of which signal recovery is difficult. This is a fundamental problem in many fields, including optics \cite{walther}, X-ray crystallography \cite{millane}, astronomical imaging \cite{dainty}, speech processing \cite{rabiner}, particle scattering, electron microscopy etc. 

Recovering a signal from its Fourier transform magnitude is known as phase retrieval. Since squared-magnitude of the Fourier transform and autocorrelation are Fourier pairs, the phase retrieval problem can be equivalently posed as recovering a signal from its autocorrelation. 

Let $\x=\{x_0,x_1,....x_{n-1}\}$ be a discrete-time signal of length $n$ and sparsity $k$, where sparsity is defined as the number of non-zero elements. Its autocorrelation, denoted by $\A=\{a_0,a_1,....a_{n-1}\}$, is defined as 
\begin{equation}
a_i \overset{def}{=} \sum_jx_jx_{j+i} = (\x \star \tilde{\x})_i
\end{equation}
where $\tilde{\x}$ is the time-reversed version of $\x$. Also, let $V$ and $W$ denote the support set of the signal $\x$ and its autocorrelation $\A$ respectively, defined as
\begin{equation}
V=\{i|x_i \neq 0\}  \quad \quad \quad \& \quad \quad \quad W=\{i|a_i \neq 0\}
\end{equation}
The phase retrieval problem can be written as
\begin{align}
\nonumber & \textrm{find} \hspace{2.5cm} \x \\
& \textrm{subject to}  \hspace{1.4cm} \x \star \tilde{\x}=\A
\end{align}

{\bf Connection to integer recovery problem:} It is often useful to be able to reconstruct the support set of the signal $V$ from the support set of the autocorrelation $W$ \cite{fienup3}. In many applications (e.g, astronomy), the signal's support set is the desired information. In other applications, support knowledge makes  signal reconstruction process using available techniques  significantly easier \cite{kishore}. For example, in \cite{kishore2} we provide an algorithm which provably recovers sparse signals uniquely from their autocorrelation if the support set is known.

We will assume that if $a_i=0$, then no two elements in $\x$ are separated by a distance $i$, i.e.,
\begin{equation}
a_i=0 \Rightarrow x_jx_{i+j}=0 \ \forall \ j 
\end{equation}
This is a very weak assumption and holds with probability one if the non-zero entries of the signal are chosen from a non-degenerate distribution. With this assumption, the support recovery problem can be posed as
\begin{align}
\nonumber & \textrm{find} \hspace{2.5cm} V \\
& \textrm{subject to}  \hspace{1.4cm} \{|i-j| ~| ~(i,j) \in V\}=W
\end{align}
Note that $V$ is a set of integers, and $W$ is exactly its pairwise distance set.

\subsection{Computational Biology}

Over the last few years, there has been a lot of interest in DNA restriction site analysis. A DNA strand is a string on the letters $\{A,T,G,C\}$. Unfortunately, the DNA string cannot be explicitly observed and in order to map it, biochemical techniques which provide indirect information have been developed.

When a particular restriction enzyme is added to a DNA solution, the DNA is cut at particular restriction sites. For example, the enzyme \textit{EcoRI} cuts at locations of the pattern GAATTC. The goal of restriction site analysis is to determine the locations of every site for a given enzyme. In order to do this, a batch of DNA is exposed to a restriction enzyme in limited quantity so that fragments of all possible lengths exist (see Figure \ref{DNA}). Using gel electrophoresis, the fragment lengths can be measured. 

\begin{figure}[H]
\begin{centering}
\includegraphics[scale=0.45]{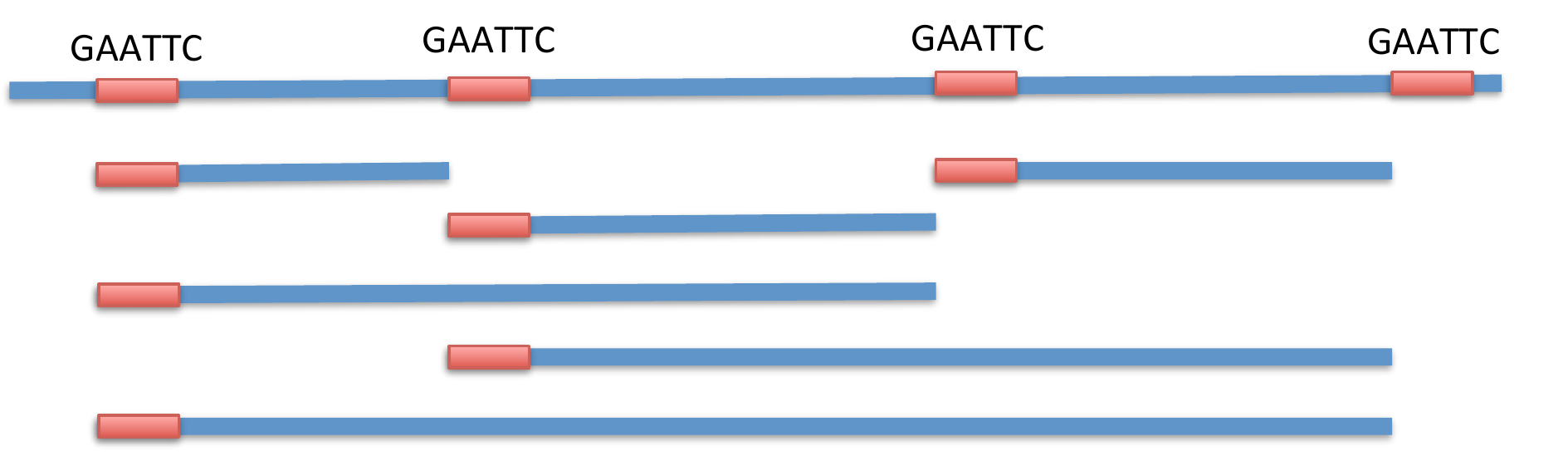}
\caption{Partial Digest Problem}
\label{DNA}
\end{centering}
\end{figure}

Recovering the locations of the restriction sites from the measured fragment lengths is known as the partial digest problem (a.k.a turnpike problem, see \cite{dakic}). The locations of the restriction sites correspond to the set of integers $V$, and the measured fragment lengths correspond to the set of pairwise distances $W$.

\section{Contributions}

Researchers have proposed a wide range of heuristics \cite{gerchberg, fienup1} to solve the phase retrieval problem, a brief summary of which can be found in \cite{fienup2}. \cite{bauschke} provides a theoretical framework to understand the heuristics, which are in essence an alternating projection between a convex set and a non-convex set. The problem with such an approach is that convergence is often to a local minimum, hence chances of successful recovery are less. Also, no theoretical guarantees can be provided.

\begin{figure}[H]
\begin{centering}
\hspace{2cm}
\includegraphics[scale=0.45]{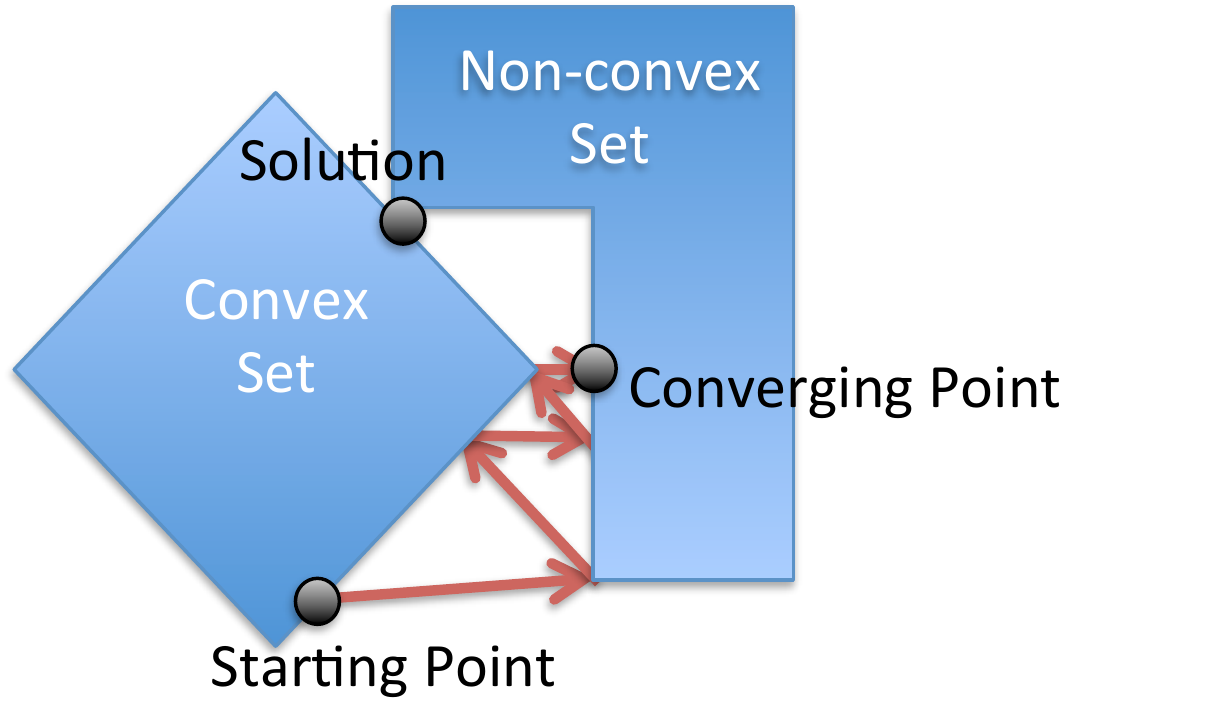}  
\caption{Alternating projection between a convex and a non-convex set}
\label{converge}
\end{centering}
\end{figure}

A variant of the partial digest problem, known as the turnpike problem, which is the problem of recovering an integer set from their pairwise distance multiset (multiplicity information of each pairwise distance also available) is also well studied. The most widely used algorithm to do this recovery is a worst-case exponential algorithm based on backtracking \cite{skiena}. \cite{dakic} provides a comprehensive summary of the existing algorithms. The question of unique provable recovery  using polynomial-time algorithms remains unanswered, and advanced equipments have been used in practice to obtain additional information to solve the problem \cite{karp,pandu}.

In many applications of these problems, the underlying signals are naturally sparse. For example, astronomical imaging deals with the locations of the stars in the sky, X-ray crystallography deals with the density of atoms and so on. In DNA restriction site analysis, it is very reasonable to assume that the restriction sites are sparsely distributed. 

Recently, some attempts have been made to exploit sparsity. An alternating projection based heuristic was proposed in \cite{vetterli}, a semidefinite relaxation based heuristic was explored in \cite{eldar}. In this work, we develop a probabilistic polynomial-time algorithm which can provably recover the underlying set uniquely if it is $O(n^{1/2-\eps})$ sparse.

\section{Main Result}

Suppose $V=\{v_0,v_1,....,v_{k-1}\}$ is a set of $k$ integers and $W=\{w_0,w_1,....,w_{K-1}\}$ is its pairwise distance set\footnote[3]{The elements of $V$ and $W$ are assumed to be in ascending order without loss of generality for convenience of notation, i.e., $v_0 < v_1 <.... < v_{k-1}$ and  $w_0 < w_1 <.... < w_{K-1}$}.

\begin{thm}[Main Result]
$V$ can be recovered uniquely (upto linear shift and reversal) from $W$ in polynomial-time with probability greater than $1-\delta$ for any $\delta>0$ if 
\begin{enumerate}[(i)]
\item $\exists ~ n \geq w_{K-1}$ such that $V$ is chosen uniformly at random\footnote{Can use other distributions too. For example, each integer belongs to the set $V$ with probability $\frac{k}{n}$ independent of each other} from $\{0,1,....,n-1\}$
\item $k=O(n^{1/2-\eps})$
\item $n > n(\eps,\delta)$
\end{enumerate}
\end{thm}

In order to overcome the trivial ambiguity of linear shift and reversal, we attempt to recover the equivalent solution set $U=\{u_0,u_1,....,u_{k-1}\}$ defined as follows
\begin{equation}
U=\begin{cases}
& V-v_0 \quad \textrm{if} \quad v_1-v_0 \leq v_{k-1}-v_{k-2}\\
& v_{k-1}-V \quad \textrm{otherwise}
\end{cases}
\end{equation}

i.e., the equivalent solution set $U$ we attempt to recover has the following properties:
\begin{enumerate}[(i)]
\item $u_0=0$ 
\item $u_1 - u_0 \leq u_{k-1}-u_{k-2}$
\end{enumerate}

\section{Theory}

Let $u_{ij}=u_j-u_i$ for $0 \leq i \leq j \leq k-1$. With this definition, $W=\{u_{ij}: 0 \leq i \leq j \leq k-1\}$ and $U=\{u_{0j}: 0 \leq j \leq k-1\}$. Note that $U \subseteq W$.

\subsection{Intersection Step}

The key idea of this step can be summarized as follows: suppose we know the value of $u_{0p}$ for some $p$, if $U_p$ and $W_p$ are defined as 
\begin{equation}
U_p=\{u_{0j}: p \leq j \leq k-1\}  \quad \&  \quad W_p=W+u_{0p}
\end{equation}
then we have
\begin{equation}
U_p \subseteq W \cap W_p 
\end{equation}

The idea can be extended to multiple intersections. Suppose we know $\{u_{0i_p}:1\leq p\leq t\}$, we can construct $\{W_{i_p}:1\leq p\leq t\}$ and see that
\begin{equation} 
U_{i_t} \subseteq W \cap \left(\bigcap_{p=1}^{t} W_{i_p}\right)
\end{equation}

\subsection{Graph Step}

For a given integer set $Z=\{z_0,z_1,....z_{|Z|-1}\}$, construct a graph $G(Z)$ with $|Z|$ vertices such that there exists an edge between $z_i$ and $z_j$ iff the following is satisfied 
\begin{equation}
\forall z_g,z_h \in Z, \quad z_g-z_h\neq z_i-z_j \quad \textrm{unless} \quad  (i,j)=(g,h) 
\end{equation}
i.e., there exists an edge between two vertices if their corresponding pairwise distance is unique. For example, consider the integer set $Z=\{2,5,19,53,67,84\}$. The graph $G(Z)$ looks as shown in Figure \ref{graph1}.

\begin{figure}[H]
\begin{centering}
\includegraphics[scale=0.6]{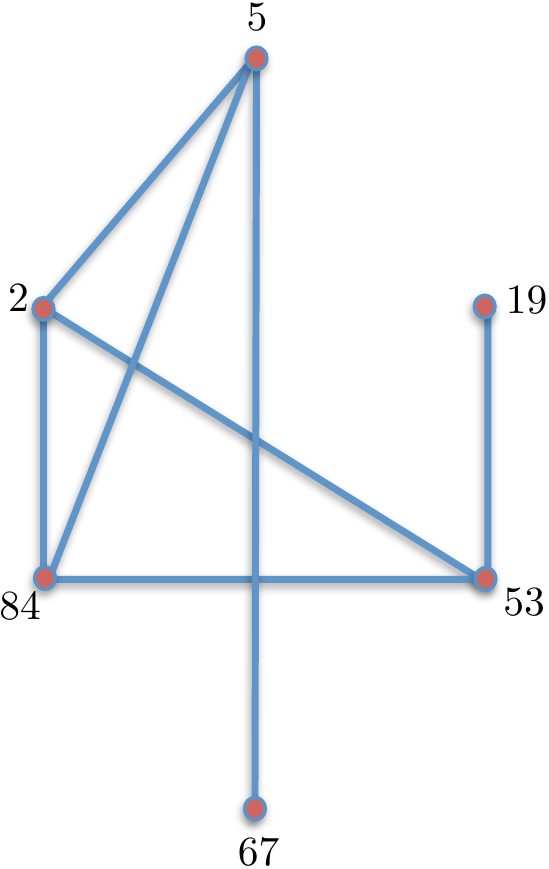}
\caption{The graph $G(Z)$ for $Z=\{2,5,19,53,67,84\}$}
\label{graph1}
\end{centering}
\end{figure}
 There exists an edge between $2$ and $5$ as it is the only pair of integers with difference $3$, there doesn't exist an edge between $5$ and $19$ because there is another pair $\{53,67\}$ with difference $14$ and so on.

The main idea of this step is as follows: suppose we construct a graph $G(Z)$ where $U \subseteq Z \subseteq W$. If there exists an edge between a pair of integers in $\{z_i,z_j\} \in Z$ such that  $z_i-z_j \in W$, then $\{z_i,z_j\} \in U$. This holds because if $\{z_i,z_j\} \notin U$, then since $z_i-z_j \in W$ there has to be a pair of integers in $U$ which have a pairwise distance $z_i-z_j$, which would contradict the fact that there is an edge between $z_i$ and $z_j$.

\section{Algorithm}

\begin{algorithm}[H]
\caption{Integer Recovery Algorithm}
\label{algo1} 
\textbf{Input:} Pairwise distance set $W$ \\
\textbf{Output:} Integer set $U$ which realizes $W$

\begin{enumerate}[1.]
\item Infer $u_{01}$ from $W$
\item Construct the set $W_1=W+u_{01}$ 
\\ \\ If $k=O(n^{1/4-\eps})$
\item Calculate $U_1=W \cap W_1$
\item Recover $U=\{0\} \cup U_1$
\\ \\ Else if $k=O(n^{1/2-\eps})$
\item Construct the graph $G(\{0\} \cup (W\cap W_1))$ and infer $\{u_{0i_p}: 1 \leq p \leq t=log(k)\}$
\item Construct the set $W_{i_p}=W+u_{0i_p}$ for $1 \leq p \leq t$
\item Calculate $U_{i_t}=W \cap \left(\bigcap_{p=1}^{t} W_{i_p}\right)$
\item Define $\tilde{U}=\{\tilde{u}_0,...\tilde{u}_{k-1}\}$ as $\tilde{U}=u_{k-1}-U$ and infer $\{\tilde{u}_{0p}$: $1 \leq p \leq t\}$ from $U_{i_t}$
\item Construct the set $\tilde{W}_{p}=W+\tilde{u}_{0p}$ for $1 \leq p \leq t=log(k)$
\item Calculate $\tilde{U}_{t}=\left(\bigcap_{p=0}^{t} \tilde{W}_p\right)$
\item Recover $\tilde{U}=\{\tilde{u}_{0p}:0\leq p \leq t-1\} \cup \tilde{U}_t$
\item Recover $U=\tilde{u}_{k-1}-\tilde{U}$
\end{enumerate}
\end{algorithm}

Algorithm \ref{algo1} outlines the steps needed to recover the set of integers from the set of their pairwise distances. Lemma \ref{u01} explains step 1. The proof for step 3 is given in Corollary \ref{4rt}. The proof for Step 5 is given in Lemma \ref{glem}.  Lemma \ref{intt} provides the proofs for steps 7 and 10.

\subsection{Circular Pairwise Distances (Beltway problem)}
In the phase retrieval setup, if the autocorrelation is defined circularly, i.e.,
\begin{equation}
a_i \overset{def}{=} \sum_jx_jx_{j+i} = (\x \star \tilde{\x})_i
\end{equation}
where the indices are considered modulo $n$, the integer recovery problem  can be stated as

\begin{align}
\nonumber & \textrm{find} \hspace{2.5cm} V \\
& \textrm{subject to}  \hspace{1.4cm} \{(i-j)\textrm{mod}(n) ~| ~(i,j) \in V\}=W
\end{align}
Algorithm \ref{algo1} can be modified to solve this problem. Suppose $u_{ij}=(u_j-u_i)\textrm{mod}(n)$ for $0 \leq i,j \leq k-1$. The intersection step in this case can be modified as follows: suppose we know the value of $u_{0p}$ for some $p$, if $W_p$ is defined as 
\begin{equation}
W_p=(W+u_{0p})\textrm{mod}(n)
\end{equation}
then we have
\begin{equation}
U \subseteq W \cap W_p 
\end{equation}
The idea can be extended to multiple intersections. Suppose we know $\{u_{0i_p}:1\leq p\leq t\}$, we can construct $\{W_{i_p}:1\leq p\leq t\}$ and see that
\begin{equation} 
U  \subseteq W \cap \left(\bigcap_{p=1}^{t} W_{i_p}\right)
\end{equation}
The graph step can be modified by having an edge between $z_i$ and $z_j$ iff the following is satisfied 
\begin{equation}
\forall z_g,z_h \in Z, \quad (z_g-z_h)\textrm{mod}(n)\neq (z_i-z_j)\textrm{mod}(n) \quad \textrm{unless} \quad  (i,j)=(g,h)
\end{equation}

Algorithm \ref{algo2} outlines the steps needed to recover the set of integers from the set of their circular pairwise distances. 
 
\begin{algorithm}[H]
\caption{Integer Recovery Algorithm}
\label{algo2} 
\textbf{Input:} Circular pairwise distance set $W$ \\
\textbf{Output:} Integer set $U$ which realizes $W$

\begin{enumerate}[1.]
\item Infer $u_{01}$ from $W$
\item Construct the set $W_1=(W+u_{01})\textrm{mod}(n)$ 
\item Calculate $W \cap W_1$, infer $u_{02}$
\item Construct the set $W_2=(W+u_{02})\textrm{mod}(n)$
\item Construct the graph $G(W\cap W_1 \cap W_2)$ and infer $\{u_{0i_p}: 1 \leq p \leq t=log(k)\}$
\item Construct the set $W_{i_p}=(W+u_{0i_p})\textrm{mod}(n)$ for $1 \leq p \leq t$
\item Calculate $U=W \cap \left(\bigcap_{p=1}^{t} W_{i_p}\right)$
\end{enumerate}
\end{algorithm}

The proofs for all the steps are similar to Algorithm 1, except for steps 1 and 3. In step 1, $u_{01}$ can be chosen as the minimum non-zero value in $W$ without loss of generality. In step 3, a choice has to be made between clockwise and anti-clockwise completion after selecting $u_{01}$, and can be done as follows: $W\cap W_1$ has terms $\{0, u_{02}\}$ and $\{u_{01}, u_{21}=u_{01}-u_{02}\}$ with pairwise distance $u_{02}$ (distances considered modulo $n$). It can be seen that if $k=O(n^{1/2-\eps})$, then $u_{02}$ is the minimum number greater than $u_{01}$ in $W$ with this property with arbitrarily high probability. Hence, $u_{02}$ can be chosen by using the minimum value $c$ in $W$ greater than $u_{01}$, such that two pairs have a pairwise distance $c$, one of them being $\{0,c\}$ and the other being $\{u_{01},u_{01}-c\}$.

\section{Proof of Main Theorem}

Suppose $V$ is a $k$-element subset of $\{0,1,....,n-1\}$ chosen uniformly at random, $k=O(n^{1/2-\eps})$ and $n>n(\eps,\delta)$.

\begin{lem} 
$u_{01}$ can be inferred from $W$.
\label{u01}
\end{lem}
\begin{proof}
The first and second highest terms in $W$ are given by $u_{0,k-1}$ and $u_{1,k-1}$ respectively since $u_{01} \leq u_{k-2,k-1}$. Since $u_{01}=u_{0,k-1}-u_{1,k-1}$, we can calculate $u_{01}$ from $W$.
\end{proof}

\begin{lem}
Probability that an integer $l$ belongs to $W$ is 
\label{int0}
\begin{enumerate}[(i)]
\item $1$ if $l \in U$.
\item less than or equal to $\frac{k^2}{n}+o(\frac{k^2}{n})$ if $l \notin U$.
\end{enumerate}
\end{lem}
\begin{proof}
\begin{enumerate}[(i)]
\item Since $0 \in U$, if $l \in U$ then $l \in W$.
\item If $l \notin U$, probability that an index $l$ belongs to $W$ can be bounded as follows
\begin{equation}
Pr\{l \in W\ | l \notin U\} = Pr\{ \bigcup_g g, g+l \in V| l \notin U\} \leq \frac{k^2}{n}+o\left(\frac{k^2}{n}\right)
\end{equation}
\end{enumerate}
\end{proof}

\begin{lem}
The probability that an integer $l$ not in $U$ belongs to $W \cap W_2$ is less than or equal to $\frac{k^4}{n^2}+o(\frac{k^4}{n^2})$
\label{int2}
\end{lem}
\begin{proof}
\begin{equation}
Pr\{l \in W\cap W_2 | l \notin U\} = Pr\{l, l-u_{01} \in W | l \notin U\} 
\end{equation}
\begin{equation}
\leq \sum_dPr\{u_{01}=d\}\left(\sum_{g,h}Pr\{g,g+l,h,h+l-d \in V| (0,d) \in U, (1,..,d-1,l)\notin U\}\right)
\end{equation}
\begin{equation}
\leq \sum_dPr\{u_{01}=d\}\left(\frac{k^4}{n^2}+o\left(\frac{k^4}{n^2}\right)\right)  = \frac{k^4}{n^2}+o\left(\frac{k^4}{n^2}\right) 
\end{equation}
\end{proof}

\begin{lem}
\label{survive}
The probability that an index $l$ in $W$, which is not in $U$, belongs to $W \cap W_2$ is less than $\frac{k^2}{n}+o(\frac{k^2}{n})$
\end{lem}
\begin{proof}
\begin{equation}
Pr\{l \in W\cap W_2 | l \in W, l \notin U\} \leq \frac{Pr\{l \in W\cap W_2| l \notin U\}}{Pr\{l \in W| l \notin U\}} \leq \frac{\frac{k^4}{n^2}+o(\frac{k^4}{n^2})}{\frac{k^2}{n}-o(\frac{k^2}{n})} \leq \frac{k^2}{n} +o\left(\frac{k^2}{n}\right)
\end{equation}
\end{proof}

\begin{lem} 
$Pr\{u_{0t} > u_{k-k^\alpha,k-1}\} \leq \delta$ for any $\alpha, \delta>0$ if $t=log(k)$.
\label{dist}
\end{lem}
\begin{proof}
Since the integers are chosen uniformly at random, we have
\begin{equation}
E[u_{0t}] = \sum_{i=0}^{i=t-1} E[u_{i,i+1}]  \leq \frac{tn}{k} \quad \& \quad  var[u_{0t}] =  E[u_{0t}^2] - E[u_{0t}]^2 \leq \frac{2tn^2}{k^2}  + o\left(\frac{2tn^2}{k^2}\right)
\end{equation}
Using Chebyshev's inequality, we get
\begin{equation}
Pr\{u_{0t} > \frac{k^{\alpha/2}n}{k}\} \leq \frac{ \frac{2tn^2}{k^2}  + o\left(\frac{2tn^2}{k^2}\right)}{\left(\frac{k^{\alpha/2}n}{k}-o\left(\frac{k^{\alpha/2}n}{k}\right)\right)^2}\leq \frac{2t}{k^\alpha}+o\left(\frac{2t}{k^\alpha}\right)
\end{equation}
Similarly, we have
\begin{equation}
E[u_{k-k^\alpha,k-1}]= \frac{(k^\alpha-1) n}{k}  \quad \quad \& \quad \quad var\{u_{k-k^\alpha,k-1}\} \leq \frac{2k^\alpha n^2}{k^2}  + o\left(\frac{2k^\alpha n^2}{k^2}\right)
\end{equation}
and
\begin{equation}
Pr\{u_{k-k^\alpha,k-1}<\frac{k^{\alpha/2} n}{k}\} \leq \frac{ \frac{2k^\alpha n^2}{k^2}  + o\left(\frac{2k^\alpha n^2}{k^2}\right)}{\left(\frac{k^{\alpha}n}{k}-o\left(\frac{k^{\alpha}n}{k}\right)\right)^2}\leq \frac{2}{k^\alpha}+o\left(\frac{2}{k^\alpha}\right)
\end{equation}
Hence, we can bound the desired probability as follows
\begin{equation}
Pr\{u_{0t} > u_{k-k^\alpha,k-1}\} \leq  \frac{2t}{k^\alpha}+o\left(\frac{2t}{k^\alpha}\right) < \delta
\end{equation}
for any $\delta>0$.
\end{proof}

\begin{lem}
In the graph $G(\{0\} \cup (W \cap W_2))$, integers $\{u_{0p}: 1 \leq p \leq t=log(k)\}$ have an edge with $u_{0,k-1}$ with probability greater than $1-\delta$ for any $\delta>0$.
\label{glem}
\end{lem}
\begin{proof}
For any fixed $p$ such that $1 \leq p \leq t$, note that terms $u_{0p}$ and $u_{0,k-1}$ have a pairwise distance $u_{p,k-1}$. For there to be no edge between $u_{0p}$ and $u_{0,k-1}$, another integer pair should have the same pairwise distance. For this to happen, atleast one of the integers should be greater than $u_{p,k-1}$. The only integers greater than $u_{p,k-1}$ in  $W$ can be  $\{u_{ij}: 0 \leq i \leq p-1,  j>i\}$. Hence, it is sufficient to show that none of these terms exist in $W_2$ with the desired probability. These terms can be split into two cases:
\begin{enumerate}[(i)]
\item $j \leq k-k^\alpha$: 
\\Using Lemma \ref{dist}, we see that this event can be bounded with probability $\delta$ for any $\delta>0$.
\begin{equation}
u_{ij} \leq u_{0t}+u_{pj} < u_{k-k^\alpha,k-1}+u_{p,k-k^\alpha} = u_{p,k-1}
\end{equation}
\item $k-k^\alpha< j \leq k$
\\There are $k^\alpha$ such terms for a given $p$. If all the $k^\alpha$ terms don't survive $W\cap W_2$, we can be assured of an edge between $u_{0p}$ and $u_{0,k-1}$.
\end{enumerate}

Hence, if a total of $tk^\alpha$ terms in $W$ don't survive $W\cap W_2$, we are through. The probability that none of them survive $W\cap W_2$ can be upper bounded by ${tk^\alpha}(\frac{k^2}{n}+o(\frac{k^2}{n}))$ using union bounds and Lemma \ref{survive}. This term can be made less than $\delta$ for any $\delta>0$.

\end{proof}

\begin{lem}
The probability that an integer $l$ not in $U$ belongs to $W \cap \left(\bigcap_{p=1}^{t} W_{i_p}\right)$ is less than or equal to $(\frac{k^{2(1+\eps)}}{n})^{\sqrt{t}/2}+o(\frac{k^{2(1+\eps)}}{n})^{\sqrt{t}/2}$ for $t = log(k)$
\label{intt}
\end{lem}
\begin{proof}
\begin{equation}
Pr\{l \in W \cap \left(\cap_{p=1}^{t} W_{i_p}\right) | l \notin U\} = Pr\{l, \{l-u_{0i_p}:1 \leq p \leq t\} \in W | l \notin U\} 
\end{equation}
\begin{equation}
\label{eqn}
= \sum_{d_1,..d_t} Pr\{u_{0i_p}=d_p: 1 \leq p \leq t\} Pr\{l, \{l-d_p:1 \leq p \leq t\} \in W | l \notin U,u_{0i_p}=d_p: 1 \leq p \leq t\}
\end{equation}
Using the proofs in \cite{kishore2}, we can show that the numbers $\{d_p:1 \leq p \leq t\}$ have unique pairwise distances with probability $1-\delta$ for any $\delta>0$.  Using this property, we see that atleast $\sqrt{t}$ additional vertices in $U$ (apart from $\{u_{0i_p}:1 \leq p \leq t\}$) and atmost  $2t$ additional vertices are needed to realize the $t$ pairwise distances. The probability can hence be bounded as
\begin{equation}
\leq \sum_{\omega=\sqrt{t}}^{2t}{(\omega+t)^t\left(\frac{k}{n}\right)^\omega}n^{\omega/2} \leq (\frac{k^{2(1+\eps)}}{n})^{\sqrt{t}/2}+o(\frac{k^{2(1+\eps)}}{n})^{\sqrt{t}/2}
\end{equation}
for $t=log(k)$.
\end{proof}

\begin{lem}
$U_{i_t}=W\cap \left(\bigcap_{p=1}^{t} W_{i_p}\right)$ with probability greater than $1-\delta$ for any $\delta>0$ if $t\geq log(k)$
\end{lem}
\begin{proof}
The expected number of terms in $W \cap \left(\bigcap_{p=1}^{t} W_{i_p}\right)$ not in $U$ can be calculated using linearity property of expectation as
\begin{equation}
E[\#l \in W\cap\left(\cap_{p=1}^{t} W_{i_p}\right)| l \notin U] \leq n\left(\left(\frac{k^{2(1+\eps)}}{n}\right)^{\sqrt{t}/2}+
o\left(\frac{k^{2(1+\eps)}}{n}\right)^{\sqrt{t}/2}\right) \leq \delta 
\end{equation}
for any $\delta>0$. Using Markov inequality, we get
\begin{equation}
Pr\{\#l \in W \cap\left(\cap_{p=1}^{t} W_{i_p}\right)| l \notin U \geq 1 \}  \leq \delta 
\end{equation}
which completes the proof.
\end{proof}

\begin{cor}
$U_2=W\cap W_2$ with probability greater than $1-\delta$ for any $\delta>0$ if $k=O(n^{1/4-\eps})$, $n>n(\eps,\delta)$.
\label{4rt}
\end{cor}

\section{Numerical Simulations}

Consider the following example:
\begin{equation*}
W=\{\textrm{\color{black}0},     2,     \textrm{\color{black}3},    \textrm{\color{black}5},    8,    12,    14,    \textrm{\color{black}17},    30,    33,    37,    38,    49,    51,
\end{equation*}
\begin{equation*}
\hspace{1cm}52,    \textrm{\color{black}54} ,   57,    60,    68,    71,   76,    89,    90 ,   94,    97,    101,       
\end{equation*}
\begin{equation*}
\hspace{0.9cm} \textrm{103, \color{black}106},108,    109 ,   \textrm{\color{black}111} ,   \textrm{\color{black}114} ,   127   ,128,   139 ,  
\end{equation*}
\begin{equation*}
\hspace{-1cm}  141  ,     \textrm{\color{black}144}, 165  ,177 ,179,  \textrm{\color{black}182}  \}
\end{equation*}
We can infer $u_{01}=182-179=3$ from $W$. Construct $W_1=W+u_{01}$ and calculate $W \cap W_1$.
\begin{equation*}
W \cap W_1= \{\textrm{\color{black}3},      \textrm{\color{black}5},    8,     \textrm{\color{black}17},  33 ,    52,       \textrm{\color{black}54},    57,    60   ,71, 
\end{equation*}
\begin{equation*}
\hspace{2cm}   97 ,  \textrm{\color{black}106},  109,   \textrm{\color{black}111}, \textrm{\color{black}114}, \textrm{\color{black}144}, \textrm{\color{black}182}\}
\end{equation*}
Construct $G(\{0\} \cup (W \cap W_1))$ to see that
\begin{equation*}
\{0\} \leftrightarrow \{5,17\}
\end{equation*}
from which we can infer $u_{02}=5$ and $u_{03}=17$. Construct $W_2=W+u_{02}$ and $W_3=W+u_{03}$  and calculate $\left(\bigcap_{p=0}^{3} W_p\right)$
\begin{equation*}
\left(\bigcap_{p=0}^{3} W_p\right) =  \{ \textrm{\color{black}54},  \textrm{\color{black}106},   \textrm{\color{black}111}, \textrm{\color{black}114}, \textrm{\color{black}144}, \textrm{\color{black}182}\}
\end{equation*}
Calculate $U=\{u_{0p}:0 \leq p \leq 3\} \bigcup\left(\bigcap_{p=0}^{3} W_p\right)$
\begin{equation*}
U= \{\textrm{\color{black}0},\textrm{\color{black}3},\textrm{\color{black}5},\textrm{\color{black}17}, \textrm{\color{black}54},  \textrm{\color{black}106},   \textrm{\color{black}111}, \textrm{\color{black}114}, \textrm{\color{black}144}, \textrm{\color{black}182}\}
\end{equation*}

Simulations were performed by choosing k-element subsets $V$ uniformly and randomly between $\{0,1,....,n-1\}$ for different values of $n$ and $k$. Figure \ref{sim1} plots the probability of successful recovery for $n=512$ and $n=1024$.

\begin{figure}[H]
\begin{centering}
\includegraphics[scale=0.75]{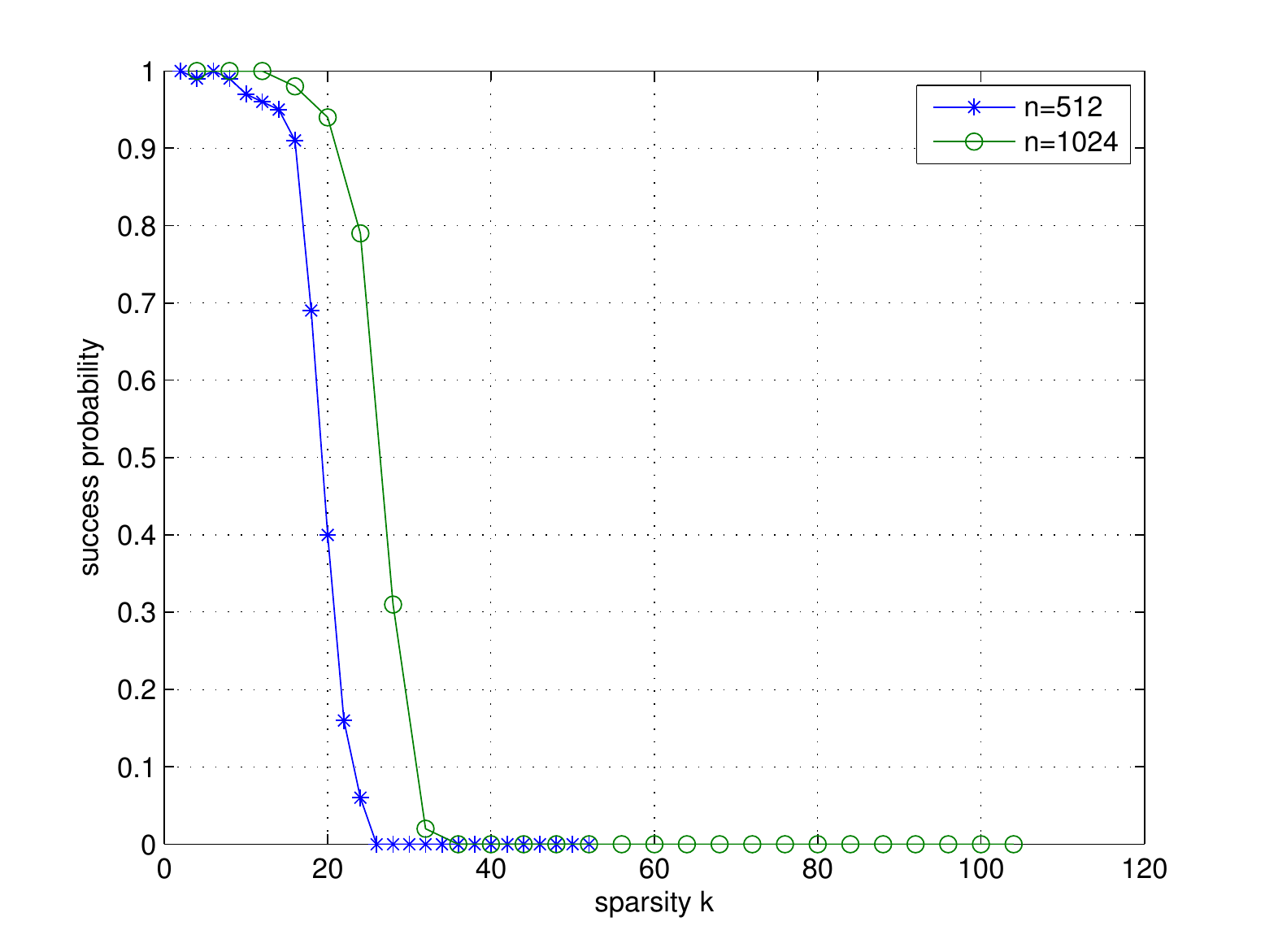}  
\caption{Probability of successful recovery}
\label{sim1}
\end{centering}
\end{figure}


\begin{thebibliography}{14}
\bibitem{patt1} A. L. Patterson: A direct method for the determination of the components of interatomic distances in crystals, Zeitschr. Krist. 90 (1935) 517-542 
\bibitem{patt2} A.L. Patterson: Ambiguities in the X-ray analysis of crystal structures, Phys. Review 65 (1944) 195-201.
\bibitem{stef}M. Stefik, ``Inferring DNA structures from segmentation data", Artificial Intelligence 11 (1978).
\bibitem{shamos} M.I. Shamos, ``Problems in computational geometry", CMU, Pittsburgh, PA 1977.
\bibitem{walther}A. Walther, "The question of phase retrieval in optics," Opt. Acta 10, 41Ð49 (1963). 
 \bibitem{millane}R. P. Millane, "Phase retrieval in crystallography and optics," J. Opt. Soc. Am. A (1990) 
\bibitem{dainty}J.C. Dainty and J.R. Fienup,``Phase Retrieval and Image Reconstruction for Astronomy,"  Chapter 7 in H. Stark, ed., Image Recovery: Theory and Application pp. 231-275.
\bibitem{rabiner}L. Rabiner and B.H. Juang, ``Fundamentals of Speech Recognition," Signal Processing Series, Prentice Hall, 1993.
\bibitem{dakic} T. Dakic, "On the Turnpike Problem", PhD Thesis, Simon Fraser University, 2000.
\bibitem{gerchberg}R. W. Gerchberg and W. O. Saxton. ``A practical algorithm for the determination of the phase from image and diffraction plane pictures''. Optik 35, 237 (1972).
\bibitem{fienup1} Fienup J.R, ``Reconstruction of an object from the modulus of its Fourier transform", Optics letters (1978)
\bibitem{fienup2}J. R. Fienup, ``Phase retrieval algorithms: a comparison''. Appl. Opt. 21 (1982).  
\bibitem{fienup3}J.R. Fienup, T.R. Crimmins, and W. Holsztynski, ``Reconstruction of the support of an object from the support of its autocorrelation", JOSA, Vol. 72, Issue 5, pp. 610-624 (1982).
\bibitem{kishore} K.Jaganathan, S. Oymak and B.Hassibi, ``Recovery of Sparse 1-D Signals from the Magnitudes of their Fourier Transform", ISIT 2012.
\bibitem{bauschke} H.H. Bauschke, P.L. Combettes and D.R. Luke, "Phase retrieval, error reduction algorithm, and Fienup variants: a view from convex optimization" J. Opt. Soc. Am.A (2002).
\bibitem{skiena} S.S. Skiena, W.D. Smith and P. Lemke, ``Reconstructing Sets from Interpoint Distances (Extended Abstract)",  SCG' 90 Proceedings of the sixth annual symposium on computational geometry, Pages 332-339.
\bibitem{karp} R.M. Karp and L.A. Newberg, `` An Algorithm for Analyzing Probed Partial Digestion Experiments", Comput Appl Biosci (1995) 11(3): 229-235 doi:10.1093/bioinformatics/11.3.229.
\bibitem{pandu} G. Pandurangan and H.Ramesh, ``The Restriction Mapping Problem Revisited", Journal of Computer and System Sciences, 2002.
\bibitem{vetterli}Y.M. Lu and M. Vetterli. ``Sparse spectral factorization: Unicity and reconstruction algorithms''. Acoustics, Speech and Signal Processing (ICASSP), 2011 IEEE International Conference on , vol., no., pp. 5976--5979, 22-27 May 2011
 \bibitem{eldar}Y. Shechtman, Y.C. Eldar, A. Szameit and M. Segev, "Sparsity Based Sub-Wavelength Imaging with Partially Incoherent Light Via Quadratic Compressed Sensing", Optics Express, vol. 19, Issue 16, pp. 14807-14822, Aug. 2011.
\bibitem{kishore2} K.Jaganathan, S. Oymak and B.Hassibi, ``On Robust Phase Retrieval for Sparse Signals", Allerton 2012.

\end{thebibliography}
\end{document}